\documentclass[graybox]{svmult}

\usepackage{mathptmx}       
\usepackage{helvet}         
\usepackage{courier}        
\usepackage{type1cm}        
\usepackage{makeidx}         
\usepackage{graphicx}        
\usepackage{multicol}        
\usepackage[bottom]{footmisc}
\usepackage{amssymb}
\usepackage{amsfonts,amsmath}

\DeclareMathOperator{\RE}{Re} \DeclareMathOperator{\IM}{Im}

\makeindex             

\begin{document}

\title*{Discrete Dirac-K\"{a}hler and Hestenes equations }
\author{Volodymyr Sushch}
\institute{Volodymyr Sushch \at Koszalin University of Technology, Sniadeckich 2,
 75-453 Koszalin, Poland, \email{volodymyr.sushch@tu.koszalin.pl}}

%
%
\maketitle

\abstract*{A discrete analogue of the Dirac equation in the Hestenes form is constructed  by introduction the Clifford product on the space of discrete forms.
We discuss the relation between the discrete   Dirac-K\"{a}hler  equation and a discrete Hestenes equation.}

\abstract {A discrete analogue of the Dirac equation in the Hestenes form is constructed  by introduction the Clifford product on the space of discrete forms.
We discuss the relation between the discrete   Dirac-K\"{a}hler  equation and a discrete Hestenes equation.}

\section{Introduction}
\label{sec:1}
The purpose of this  paper is to discuss the relation between the discrete   Dirac-K\"{a}hler  equation which was constructed in \cite{S1, S2}, and a discrete analogue of the Hestenes equation. We show that the geometric discretization scheme as developed in \cite{S2} can be used to find a new discrete formulation of the Dirac equation for a free electron in the Hestenes form.

We first briefly review some definitions and basic notation on the Dirac-K\"{a}hler equation \cite{Kahler, Rabin}.
Let $M={\mathbb R}^{1,3}$ be  Minkowski space with  metric signature  $(+,-,-,-)$.
Denote by $\Lambda^r(M)$ the vector space of smooth differential $r$-forms, $r=0,1,2,3,4$. We consider  $\Lambda^r(M)$ over $\mathbb{C}$.
Let $d:\Lambda^r(M)\rightarrow\Lambda^{r+1}(M)$ be the exterior differential and let $\delta:\Lambda^r(M)\rightarrow\Lambda^{r-1}(M)$ be the formal adjoint of $d$  with respect to  the natural inner product in $\Lambda^r(M)$ (codifferential). We have $\delta=\ast d\ast$, where  $\ast$ is the Hodge star operator  $\ast:\Lambda^r(M)\rightarrow\Lambda^{4-r}(M)$ with respect to the Lorentz metric.
 Denote by $\Lambda(M)$ the set of all differential forms on $M$. We have
\begin{equation*}
\Lambda(M)=\Lambda^0(M)\oplus\Lambda^1(M)\oplus\Lambda^2(M)\oplus\Lambda^3(M)\oplus\Lambda^4(M)=\Lambda^{ev}(M)\oplus\Lambda^{od}(M),
\end{equation*}
where $\Lambda^{ev}(M)=\Lambda^0(M)\oplus\Lambda^2(M)\oplus\Lambda^4(M)$  and  $\Lambda^{od}(M)=\Lambda^1(M)\oplus\Lambda^3(M)$.

Let $\Omega\in\Lambda(M)$
be an inhomogeneous differential form, then
$\Omega=\sum_{r=0}^4\overset{r}{\omega},$
where $\overset{r}{\omega}\in\Lambda^r(M)$.
Denote by $\Omega^{ev}$ and by $\Omega^{od}$  the even and odd parts of $\Omega$, i.e. $\Omega^{ev}=\overset{0}{\omega}+\overset{2}{\omega}+\overset{4}{\omega}$ and   $\Omega^{od}=\overset{1}{\omega}+\overset{3}{\omega}$.
 The Dirac-K\"{a}hler equation is given by
\begin{equation}\label{eq:01}
i(d+\delta)\Omega=m\Omega,
\end{equation}
where $i$ is the usual complex unit ($i^2=-1$) and  $m$  is a mass parameter.
It is easy to show that Eq.~(\ref{eq:01}) is equivalent to the set of equations
\begin{eqnarray*}\label{}
i(d+\delta)\Omega^{od}=m\Omega^{ev}, \qquad
i(d+\delta)\Omega^{ev}=m\Omega^{od}.
\end{eqnarray*}
The operator $d+\delta$ is the analogue of the gradient operator in Minkowski space-time $\nabla=\sum_{\mu=0}^3\gamma_\mu\partial^\mu$, $\mu=0,1,2,3$,  where $\gamma_\mu$ is the Dirac gamma matrix. Think of $\{\gamma_0, \gamma_1, \gamma_2, \gamma_3\}$ as a vector basis in space-time. Then the  gamma matrices $\gamma_\mu$ can be considered as generators of the Clifford algebra of space-time $\emph{C}\ell(1,3)$  \cite{B1}. The complex Clifford algebra $\emph{C}\ell(1,3)$ is a complex 16-dimensional vector space. It is known that an inhomogeneous form $\Omega$ can be represented as element of $\emph{C}\ell(1,3)$ over the complex field $\mathbb{C}$. Then the Dirac-K\"{a}hler equation can be written as an algebraic equation in $\emph{C}\ell(1,3)$ over $\mathbb{C}$
\begin{equation}\label{eq:02}
 i\nabla\Omega=m\Omega, \quad \Omega\in\emph{C}\ell(1,3).
 \end{equation}
  Eq.~(\ref{eq:02}) is equivalent to the four Dirac equations (traditional column-spinor equations) for a free electron.  Let $\emph{C}\ell^{ev}(1,3)$ be the even subalgebra of the real algebra $\emph{C}\ell(1,3)$. The equation
 \begin{equation}\label{eq:03}
 -\nabla\Omega\gamma_1\gamma_2=m\Omega\gamma_0, \quad \Omega\in\emph{C}\ell^{ev}(1,3)
 \end{equation}
is called the Hestenes form of the Dirac equation \cite{H1, H2}. The Hestenes equation is equivalent to the Dirac equation \cite{H1, Marchuk}. Suppose that for exterior forms (elements of $\Lambda(M)$) the Clifford multiplication is defined. In this case the basis covectors $e^\mu=dx^\mu$, $\mu=0,1,2,3$, of space-time are considered as generators of the Clifford algebra. The resulting algebra $\Lambda(M)$ with two multiplications is called the Grassmann-Clifford bialgebra \cite{Marchuk}. Thus Eq~(\ref{eq:03}) can be rewritten in terms of inhomogeneous forms as
\begin{equation}\label{eq:04}
 -(d+\delta)\Omega e^1e^2=m\Omega e^0,
 \end{equation}
 where $\Omega\in\Lambda^{ev}(M)$ is a real-valued  form.

 In this paper we construct a discrete analogue of the Hestenes equation (\ref{eq:04}) by introduction the Clifford product on the space of discrete forms.
In much the same way as in the continuum case \cite{B} it is shown that a solution of  the discrete   Dirac-K\"{a}hler equation gives rise to four independent solutions of  the discrete Hestenes equation.
Note that the discrete model is expressed clearly in terms of difference equations.
\section{Discrete Dirac-K\"{a}hler equation}
\label{sec:2}
We use a  discretization scheme  based on the language of differential forms and the double complex construction which is described in our preceding paper \cite{S2}.
Due to space limitations this paper does not include the relevant material from  \cite{S2}. We refer the reader to  \cite{S1, S2} for full mathematical details of the approach. This approach was originated by Dezin \cite{Dezin}.
Let $K(4)=K\otimes K\otimes K\otimes K$
be a cochain complex with  complex  coefficients,
where  $K$ is  the 1-dimensional complex generated by 0- and 1-dimensional basis elements   $x^{\kappa}$  and $e^{\kappa}$,  $\kappa\in\mathbb{Z}$,  respectively.
Then an arbitrary r-dimensional basis element of $K(4)$ can be written as
$s^k_{(r)}=s^{k_0}\otimes s^{k_1}\otimes s^{k_2}\otimes s^{k_3}$, where
$s^{k_\mu}$ is either $x^{k_\mu}$ or $e^{k_\mu}$,  $k=(k_0,k_1,k_2,k_3)$ and \ $k_\mu\in\mathbb{Z}$.
The dimension $r$ of a basis element $s^k_{(r)}$ is given
by the number of factors $e^{k_\mu}$ that appear in it. For example, the 1-dimensional basis elements
of $K(4)$ can be written as
\begin{eqnarray*}
e^k_0=e^{k_0}\otimes x^{k_1}\otimes x^{k_2}\otimes x^{k_3},  \qquad
e^k_1=x^{k_0}\otimes e^{k_1}\otimes x^{k_2}\otimes x^{k_3}, \\
e^k_2=x^{k_0}\otimes x^{k_1}\otimes e^{k_2}\otimes x^{k_3},  \qquad
e^k_3=x^{k_0}\otimes x^{k_1}\otimes x^{k_2}\otimes e^{k_3},
\end{eqnarray*}
where  the subscript $\mu=0,1,2,3$ indicates  a place of $e^{k_\mu}$ in $e^k_\mu$.
The complex $K(4)$ is a discrete analogue of $\Lambda(M)$. We will call cochains forms,
emphasizing their relationship with the corresponding continuum objects, differential
forms. Denote by  $K^r(4)$ the set of all $r$-forms. Then we have
\begin{equation*}
K(4)=K^0(4)\oplus K^1(4)\oplus K^2(4)\oplus K^3(4)\oplus K^4(4)=K^{ev}(4)\oplus K^{od}(4),
\end{equation*}
 where $K^{ev}(4)=K^0(4)\oplus K^2(4)\oplus K^4(4)$ and $K^{od}(4)=K^1(4)\oplus K^3(4)$.
 Any $r$-form $\overset{r}{\omega}\in K^r(4)$ can be expressed as
\begin{eqnarray}\label{eq:05}
\overset{0}{\omega}=\sum_k\overset{0}{\omega}_kx^k,  \qquad  \overset{4}{\omega}=\sum_k\overset{4}{\omega}_ke^k,
\end{eqnarray}
where $x^k=x^{k_0}\otimes x^{k_1}\otimes x^{k_2}\otimes x^{k_3}$  and  $e^k=e^{k_0}\otimes e^{k_1}\otimes e^{k_2}\otimes e^{k_3}$, and
\begin{eqnarray}\label{eq:06}
\overset{1}{\omega}=\sum_k\sum_{\mu=0}^3\omega_k^\mu e_\mu^k, \qquad
\overset{2}{\omega}=\sum_k\sum_{\mu<\nu} \omega_k^{\mu\nu}e_{\mu\nu}^k, \qquad
\overset{3}{\omega}=\sum_k\sum_{\iota<\mu<\nu} \omega_k^{\iota\mu\nu}e_{\iota\mu\nu}^k,
\end{eqnarray}
where $e_\mu^k$, $e_{\mu\nu}^k$ and $e_{\iota\mu\nu}^k$ are 1-, 2- and 3-dimensional basic elements of $K(4)$.
The components $\overset{0}{\omega}_k, \ \overset{4}{\omega}_k, \  \omega_k^\mu, \ \omega_k^{\mu\nu}$ and $\omega_k^{\iota\mu\nu}$ are complex numbers.
A discrete inhomogeneous form $\Omega\in K(4)$  is defined to be
\begin{equation}\label{eq:07}
\Omega=\overset{0}{\omega}+\overset{1}{\omega}+\overset{2}{\omega}+\overset{3}{\omega}+\overset{4}{\omega}.
\end{equation}
Let $d^c: K^r(4)\rightarrow K^{r+1}(4)$ be a discrete analogue of the exterior derivative $d$ and let $\delta ^c: K^r(4)\rightarrow K^{r-1}(4)$ be a discrete analogue of the codifferential $\delta$. For definitions of these operators and other discrete operations (the $\cup$-multiplication, the discrete Hodge star and so on) we refer the reader to  \cite{S2}. In this paper  we give only the difference
expression for $d^c$ and  $\delta ^c$.
Let the difference operator $\Delta_\mu$ be defined by
\begin{equation}\label{eq:08}
\Delta_\mu\omega_k^{(r)}=\omega_{\tau_\mu k}^{(r)}-\omega_k^{(r)},
\end{equation}
where  $\omega_k^{(r)}\in\mathbb{C}$ is a component of $\overset{r}{\omega}\in K^r(4)$ and
$\tau_\mu$ is   the shift operator  which acts  as
$\tau_\mu k=(k_0,...k_\mu+1,...k_3), \   \mu=0,1,2,3.$
For forms (\ref{eq:05}),  (\ref{eq:06})  we have
\begin{eqnarray}\label{eq:09}
d^c\overset{0}{\omega}=\sum_k\sum_{\mu=0}^3(\Delta_\mu\overset{0}{\omega}_k)e_\mu^k,  \qquad d^c\overset{1}{\omega}=\sum_k\sum_{\mu<\nu}(\Delta_\mu\omega_k^\nu-\Delta_\nu\omega_k^\mu)e_{\mu\nu}^k,
\end{eqnarray}
\begin{eqnarray}\label{eq:10}
d^c\overset{2}{\omega}=\sum_k\big[(\Delta_0\omega_k^{12}-\Delta_1\omega_k^{02}+\Delta_2\omega_k^{01})e_{012}^k
+(\Delta_0\omega_k^{13}-\Delta_1\omega_k^{03}+\Delta_3\omega_k^{01})e_{013}^k \nonumber \\
+(\Delta_0\omega_k^{23}-\Delta_2\omega_k^{03}+\Delta_3\omega_k^{02})e_{023}^k
+(\Delta_1\omega_k^{23}-\Delta_2\omega_k^{13}+\Delta_3\omega_k^{12})e_{123}^k\big],
\end{eqnarray}
\begin{equation}\label{eq:11}
d^c\overset{3}{\omega}=\sum_k(\Delta_0\omega_k^{123}-\Delta_1\omega_k^{023}+\Delta_2\omega_k^{013}-\Delta_3\omega_k^{012})e^k, \qquad d^c\overset{4}{\omega}=0,
\end{equation}
\begin{equation}\label{eq:12}
\delta^c\overset{0}{\omega}=0, \qquad \delta^c\overset{1}{\omega}=\sum_k(\Delta_0\omega_k^{0}-\Delta_1\omega_k^{1}-\Delta_2\omega_k^{2}-\Delta_3\omega_k^{3})x^k,
\end{equation}
\begin{eqnarray}\label{eq:13} \nonumber
\delta^c\overset{2}{\omega}=\sum_k\big[(\Delta_1\omega_k^{01}+\Delta_2\omega_k^{02}+\Delta_3\omega_k^{03})e_{0}^k
+(\Delta_0\omega_k^{01}+\Delta_2\omega_k^{12}+\Delta_3\omega_k^{13})e_{1}^k\\
+(\Delta_0\omega_k^{02}-\Delta_1\omega_k^{12}+\Delta_3\omega_k^{23})e_{2}^k
+(\Delta_0\omega_k^{03}-\Delta_1\omega_k^{13}-\Delta_2\omega_k^{23})e_{3}^k\big],
\end{eqnarray}
\begin{eqnarray}\label{eq:14} \nonumber
\delta^c\overset{3}{\omega}=\sum_k\big[(-\Delta_2\omega_k^{012}-\Delta_3\omega_k^{013})e_{01}^k+
(\Delta_1\omega_k^{012}-\Delta_3\omega_k^{023})e_{02}^k\\ \nonumber
+(\Delta_1\omega_k^{013}+\Delta_2\omega_k^{023})e_{03}^k
+(\Delta_0\omega_k^{012}-\Delta_3\omega_k^{123})e_{12}^k\\
+(\Delta_0\omega_k^{013}+\Delta_2\omega_k^{123})e_{13}^k
+(\Delta_0\omega_k^{023}-\Delta_1\omega_k^{123})e_{23}^k\big],
\end{eqnarray}
\begin{eqnarray}\label{eq:15}
\delta^c\overset{4}{\omega}=\sum_k\big[(\Delta_3\overset{4}{\omega}_k)e_{012}^k-(\Delta_2\overset{4}{\omega}_k)e_{013}^k
+(\Delta_1\overset{4}{\omega}_k)e_{023}^k+(\Delta_0\overset{4}{\omega}_k)e_{123}^k\big].
\end{eqnarray}
Let $\Omega\in K(4)$ be given by (\ref{eq:07}). A discrete analogue of the Dirac-K\"{a}hler equation (\ref{eq:01}) can be defined as
 \begin{equation}\label{eq:16}
i(d^c+\delta^c)\Omega=m\Omega.
\end{equation}
We can write this equation more explicitly by separating its homogeneous components as
\begin{eqnarray}\label{eq:17} \nonumber
i\delta^c\overset{1}{\omega}=m\overset{0}{\omega}, \quad i(d^c\overset{1}{\omega}+\delta^c\overset{3}{\omega})=m\overset{2}{\omega}, \quad
id^c\overset{3}{\omega}=m\overset{4}{\omega},\\
i(d^c\overset{0}{\omega}+\delta^c\overset{2}{\omega})=m\overset{1}{\omega}, \qquad
i(d^c\overset{2}{\omega}+\delta^c\overset{4}{\omega})=m\overset{3}{\omega}.
\end{eqnarray}
Substituting (\ref{eq:09})--(\ref{eq:15})  into (\ref{eq:17}) one  obtains the set of 16 difference equations   \cite{S2}.
\section{Clifford multiplication in $K(4)$ and discrete Hestenes equation}
\label{sec:3}
Let us define  the Clifford multiplication in  $K(4)$ by the following rules:

\medskip

1) \  $x^ke^k_\mu=e^k_\mu x^k=e^k_\mu, \quad \mu=0,1,2,3$;

2) \  $e^k_\mu e^k_\nu+e^k_\nu e^k_\mu=2g_{\mu\nu}x^k$, where  $g_{\mu\nu}=diag(1,-1,-1,-1)$ is the metric tensor;

3) \ $e^k_{\mu_1}\cdots e^k_{\mu_s}=e^k_{\mu_1\cdots \mu_s}$ \ for \ $0\leq \mu_1<\cdots <\mu_s\leq 3$.

\medskip
Note that the multiplication is defined for the basis elements of $K(4)$ with the same multi-index $k=(k_0,k_1,k_2,k_3)$ supposing the product to be zero in all other cases.
The operation is linearly extended to arbitrary discrete forms. For example, for any  $\overset{1}{\omega}, \ \overset{1}{\varphi}\in K^1(4)$ we have
\begin{eqnarray*}
\overset{1}{\omega}\overset{1}{\varphi}=\big(\sum_k\sum_{\mu=0}^3\omega^\mu_ke_\mu^k\big)\big(\sum_k\sum_{\mu=0}^3\varphi^\mu_ke_\mu^k\big)=
\sum_k (\omega^0_k\varphi^0_k-\omega^1_k\varphi^1_k-\omega^2_k\varphi^2_k-\omega^3_k\varphi^3_k)x^k \\
+\sum_k [(\omega^0_k\varphi^1_k-\omega^1_k\varphi^0_k)e_{01}^k+(\omega^0_k\varphi^2_k-\omega^2_k\varphi^0_k)e_{02}^k+
(\omega^0_k\varphi^3_k-\omega^3_k\varphi^0_k)e_{03}^k\\
+(\omega^1_k\varphi^2_k-\omega^2_k\varphi^1_k)e_{12}^k+(\omega^1_k\varphi^3_k-\omega^3_k\varphi^1_k)e_{13}^k+(\omega^2_k\varphi^3_k-\omega^3_k\varphi^2_k)e_{23}^k].
\end{eqnarray*}
\begin{proposition}
For any inhomogeneous form $\Omega\in K(4)$ we have
\begin{equation}\label{eq:18}
(d^c+\delta^c)\Omega=\sum_{\mu=0}^3e_\mu\Delta_\mu\Omega,
\end{equation}
where
\begin{equation}\label{eq:19}
e_\mu=\sum_ke_\mu^k, \qquad \mu=0,1,2,3,
\end{equation}
and $\Delta_\mu$ is the difference operator which acts on each component of $\Omega$ by the rule~(\ref{eq:08}).
\end{proposition}
\begin{proof}
We prove the claim only for the case of even forms. Similar calculations apply to the case of odd forms.
Let $\Omega^{ev}=\overset{0}{\omega}+\overset{2}{\omega}+\overset{4}{\omega}$ be the even part of $\Omega$. We have
\begin{eqnarray*}
\sum_{\mu=0}^3e_\mu\Delta_\mu\overset{0}{\omega}=\sum_k(\Delta_0\overset{0}{\omega}_ke_0^k+\Delta_1\overset{0}{\omega}_ke_1^k+\Delta_2\overset{0}{\omega}_ke_2^k+
\Delta_3\overset{0}{\omega}_ke_3^k),
\end{eqnarray*}
\begin{eqnarray*}
\sum_{\mu=0}^3e_\mu\Delta_\mu\overset{2}{\omega}=\sum_k[(\Delta_1\omega_k^{01}+\Delta_2\omega_k^{02}+\Delta_3\omega_k^{03})e_0^k+
(\Delta_0\omega_k^{01}+\Delta_2\omega_k^{12}+\Delta_3\omega_k^{13})e_1^k\\
+(\Delta_0\omega_k^{02}-\Delta_1\omega_k^{12}+\Delta_3\omega_k^{23})e_2^k+
(\Delta_0\omega_k^{03}-\Delta_1\omega_k^{13}-\Delta_2\omega_k^{23})e_3^k]\\
+\sum_k[(\Delta_0\omega_k^{12}-\Delta_1\omega_k^{02}+\Delta_2\omega_k^{01})e_{012}^k+
(\Delta_0\omega_k^{13}-\Delta_1\omega_k^{03}+\Delta_3\omega_k^{01})e_{013}^k\\
+(\Delta_0\omega_k^{23}-\Delta_2\omega_k^{03}+\Delta_3\omega_k^{02})e_{023}^k+(\Delta_1\omega_k^{23}-\Delta_2\omega_k^{13}+\Delta_3\omega_k^{12})e_{123}^k],
\end{eqnarray*}
\begin{eqnarray*}
\sum_{\mu=0}^3e_\mu\Delta_\mu\overset{4}{\omega}=\sum_k(\Delta_0\overset{4}{\omega}_ke_{123}^k+\Delta_1\overset{4}{\omega}_ke_{023}^k-\Delta_2\overset{4}{\omega}_ke_{013}^k+
\Delta_3\overset{4}{\omega}_ke_{012}^k).
\end{eqnarray*}
Summing  both sides of  the above and using (\ref{eq:09})--(\ref{eq:15}) we obtain
\begin{eqnarray*}
\sum_{\mu=0}^3e_\mu\Delta_\mu\Omega^{ev}=\sum_{\mu=0}^3e_\mu\Delta_\mu(\overset{0}{\omega}+\overset{2}{\omega}+\overset{4}{\omega})=
d^c(\overset{0}{\omega}+\overset{2}{\omega})+\delta^c(\overset{2}{\omega}+\overset{4}{\omega})=(d^c+\delta^c)\Omega^{ev}.
\end{eqnarray*}
\end{proof}
Thus the discrete Dirac-K\"{a}hler equation can be rewritten in the form
\begin{equation*}
i\sum_{i=0}^3e_\mu\Delta_\mu\Omega=m\Omega.
\end{equation*}
Let $\Omega\in K^{ev}(4)$ be a real-valued even inhomogeneous form.  A discrete analogue of the Hestenes equation (\ref{eq:04}) is defined by
\begin{equation}\label{eq:20}
-(d^c+\delta^c)\Omega e_1e_2=m\Omega e_0,
\end{equation}
where $e_0, e_1, e_2$ are given by  (\ref{eq:19}).
This equation  can be expressed in terms of difference equations.
Substituting (\ref{eq:09})--(\ref{eq:15}) into  (\ref{eq:20}), and by the rules 1)--3) we obtain
\begin{eqnarray*}\label{}
\Delta_0\omega_k^{12}-\Delta_1\omega_k^{02}+\Delta_2\omega_k^{01}+\Delta_3\overset{4}{\omega}_k=m\overset{0}{\omega}_k,\\
\Delta_2\overset{0}{\omega}_k+\Delta_0\omega_k^{02}-\Delta_1\omega_k^{12}+\Delta_3\omega_k^{23}=m\omega_k^{01},\\
-\Delta_1\overset{0}{\omega}_k-\Delta_0\omega_k^{01}-\Delta_2\omega_k^{12}-\Delta_3\omega_k^{13}=m\omega_k^{02},\\
-\Delta_1\omega_k^{23}+\Delta_2\omega_k^{13}-\Delta_3\omega_k^{12}-\Delta_0\overset{4}{\omega}_k=m\omega_k^{03},\\
-\Delta_0\overset{0}{\omega}_k-\Delta_1\omega_k^{01}-\Delta_2\omega_k^{02}-\Delta_3\omega_k^{03}=m\omega_k^{12},\\
-\Delta_0\omega_k^{23}+\Delta_2\omega_k^{03}-\Delta_3\omega_k^{02}-\Delta_1\overset{4}{\omega}_k=m\omega_k^{13},\\
\Delta_0\omega_k^{13}-\Delta_1\omega_k^{03}+\Delta_3\omega_k^{01}-\Delta_2\overset{4}{\omega}_k=m\omega_k^{23},\\
\Delta_3\overset{0}{\omega}_k+\Delta_0\omega_k^{03}-\Delta_1\omega_k^{13}-\Delta_2\omega_k^{23}=m\overset{4}{\omega}_k.
\end{eqnarray*}
Let us introduce  the following constant forms
\begin{equation}\label{eq:21}
P_{\pm 0}=\frac{1}{2}(x\pm e_0), \qquad  P_{\pm 12}=\frac{1}{2}(x\pm ie_1e_2),
\end{equation}
where  $x=\sum_kx^k$ is the unit 0-form,
 and $e_\mu$ is given by (\ref{eq:19}). Note that $x$ plays  a role of the unit element in $K(4)$.   It is easy to check that
 \begin{equation*}
(P_{\pm 0})^2=P_{\pm 0}P_{\pm 0}=P_{\pm 0}, \qquad  (P_{\pm 12})^2=P_{\pm 12}P_{\pm 12}=P_{\pm 12}.
\end{equation*}
 Hence, the forms $P_{\pm 0}$ and $P_{\pm 12}$ are projectors.
\begin{proposition}
The projectors $P_{\pm 0}$ and $P_{\pm 12}$ have the following properties
\begin{equation}\label{eq:22}
P_{\pm 0}P_{\pm 12}=P_{\pm 12}P_{\pm 0}, \quad e_0P_{\pm 0}=P_{\pm 0}e_0,  \quad e_1e_2P_{\pm 12}=P_{\pm 12}e_1e_2,
\end{equation}
\begin{equation}\label{eq:23}
P_{\pm 0}=\pm P_{\pm 0}e_0,  \qquad  P_{\pm 12}=\pm iP_{\pm 12}e_1e_2.
\end{equation}
\end{proposition}
\begin{proof}
The proof is a computation.
 \end{proof}
Let
\begin{equation}\label{eq:24}
P_{++}=P_{+0}P_{+12}, \quad P_{+-}=P_{+0}P_{-12}, \quad P_{-+}=P_{-0}P_{+12}, \quad P_{--}=P_{-0}P_{-12}.
\end{equation}
It is obvious that (\ref{eq:24}) are projectors again.
\begin{proposition}
Any inhomogeneous form $\Omega\in K(4)$ decomposes into four parts
\begin{equation}\label{eq:25}
\Omega=\Omega P_{++}+\Omega P_{-+}+\Omega P_{+-}+\Omega P_{--}.
\end{equation}
\end{proposition}
\begin{proof} By (\ref{eq:21}) $\Omega$ can be represented as
\begin{equation*}
\Omega=\Omega P_{+0}+\Omega P_{-0} \quad \mbox{or}  \quad \Omega=\Omega P_{+12}+\Omega P_{-12}.
\end{equation*}
This yields
\begin{equation*}
\Omega=(\Omega P_{+0}+\Omega P_{-0})P_{+12}+(\Omega P_{+0}+\Omega P_{-0})P_{-12}.
\end{equation*}
Hence, by (\ref{eq:24}) we obtain (\ref{eq:25}).
 \end{proof}
 Recall that the  Hestenes equation is defined on real-valued even  forms.
 First suppose that the discrete Hestenes equation (\ref{eq:20}) acts in $K(4)$, i.e. acts in the same space as  the discrete Dirac-K\"{a}hler equation.
 \begin{proposition}
Let $\Omega\in K(4)$ be a solution of the discrete Dirac-K\"{a}hler equation, then  $\Omega P_{++}$ and $\Omega P_{--}$ satisfy Eq.~(\ref{eq:20}) while
$\Omega P_{-+}$ and $\Omega P_{+-}$ satisfy the same equation but the sign of the right hand side changed to its opposite.
\end{proposition}
\begin{proof}
It suffices to prove the claim for one of projectors  (\ref{eq:24}), say for $P_{++}$.
The other cases are similar.
Multiplying Eq.~(\ref{eq:16}) from the right by the projector $P_{++}$ we obtain
\begin{equation}\label{eq:26}
i(d^c+\delta^c)\Omega P_{++}=m\Omega P_{++}.
\end{equation}
Since $P_{++}$ is constant, using (\ref{eq:22}) and (\ref{eq:23}), we have
\begin{align*}
i(d^c+\delta^c)\Omega P_{++}&=i(d^c+\delta^c)\Omega P_{+0}P_{+12}=i^2(d^c+\delta^c)\Omega P_{+0}P_{+12}e_1e_2\\
&=-(d^c+\delta^c)(\Omega P_{++})e_1e_2,
\end{align*}
\begin{equation*}
\Omega P_{++}=\Omega P_{+0}P_{+12}=\Omega P_{+12}P_{+0}e_0=\Omega P_{++}e_0.
\end{equation*}
Substituting this into (\ref{eq:26}) yields
\begin{equation*}
-(d^c+\delta^c)(\Omega P_{++})e_1e_2=m(\Omega P_{++})e_0.
\end{equation*}
 \end{proof}
Let $\overline{\Omega}$  be the complex conjugate of $\Omega$.
 Consider the real-valued forms $\Omega_+$  and $\Omega_-$ given by
\begin{equation}\label{eq:27}
\Omega_{\pm}=\pm\frac{1}{2}(\Omega+\overline{\Omega})e_0\pm\frac{i}{2}(\Omega-\overline{\Omega})e_1e_2.
\end{equation}
By (\ref{eq:22}) and (\ref{eq:23}) it is easy to check that
\begin{equation*}
\Omega P_{++}=\Omega_+P_{++},  \qquad \Omega P_{--}=\Omega_-P_{--}.
\end{equation*}
Hence, if $\Omega$ is a solution of the discrete Dirac-K\"{a}hler equation, then $\Omega_+P_{++}$ and $\Omega_-P_{--}$ are solutions of Eq. (\ref{eq:20}).
The forms $\Omega_+P_{++}$ and $\Omega_-P_{--}$ are complex-valued again.
However, if $\Omega_+P_{++}$ and $\Omega_-P_{--}$ are solutions of Eq. (\ref{eq:20}) than the real and image parts of these complex-valued forms are also solutions of Eq. (\ref{eq:20}). This is obvious since the discrete Hestenes equation is real and linear. The real and  image parts of $\Omega_+P_{++}$ are
\begin{equation*}
\RE(\Omega_+ P_{++})=\frac{1}{4}(\Omega_++\Omega_+e_0), \qquad
\IM(\Omega_+ P_{++})=\frac{1}{4}(\Omega_+e_1e_2+\Omega_+e_0e_1e_2).
\end{equation*}
Set
\begin{equation*}
\Omega_1 =\Omega_+,  \qquad \Omega_2=\Omega_+e_0, \qquad \Omega_3=\Omega_+e_1e_2, \qquad \Omega_4=\Omega_+e_0e_1e_2.
\end{equation*}
Now we take the even part of these forms. A direct  computation gives
\begin{eqnarray}\nonumber
\Omega_1^{ev}=&\frac{1}{2}(\Omega^{od}+\overline{\Omega}^{od})e_0+\frac{i}{2}(\Omega^{ev}-\overline{\Omega}^{ev})e_1e_2,
 \\ \nonumber
  \Omega_2^{ev}=&\frac{1}{2}(\Omega^{ev}+\overline{\Omega}^{ev})+\frac{i}{2}(\Omega^{od}-\overline{\Omega}^{od})e_0e_1e_2, \\
  \nonumber
  \Omega_3^{ev}=&\frac{1}{2}(\Omega^{od}+\overline{\Omega}^{od})e_0e_1e_2-\frac{i}{2}(\Omega^{ev}-\overline{\Omega}^{ev}), \\
 \Omega_4^{ev}=&\frac{1}{2}(\Omega^{ev}+\overline{\Omega}^{ev})e_1e_2-\frac{i}{2}(\Omega^{od}-\overline{\Omega}^{od})e_0,\label{eq:28}
 \end{eqnarray}
where $\Omega^{ev}$ and $\Omega^{od}$ are the even and odd parts of $\Omega=\Omega^{ev}+\Omega^{od}$.

Thus, we have proved the following
\begin{proposition}
Let $\Omega\in K(4)$ be a solution of the discrete Dirac-K\"{a}hler equation. Then  $\Omega_j^{ev}\in K^{ev}(4)$, $j=1,2,3,4$, in the form (\ref{eq:28}) are four independent solutions of the discrete Hestenes  equation (\ref{eq:20}).
\end{proposition}
It should be noted that  taking $\Omega_-P_{--}$ instead $\Omega_+P_{++}$ we also obtain  four independent solutions of Eq. (\ref{eq:20}) in the same form (\ref{eq:28}).


\begin{thebibliography}{99.}


\bibitem{B} Baylis, W.~E. ed.: Clifford (Geometric) Algebra with Applications to
Physics, Mathematics, and Engineering.  Birkh\"{e}auser, Boston (1996)

\bibitem{B1} Baylis, W.~E.: Comment on `Dirac theory in spacetime algebra'.  J. Phys. A: Math. Gen.   \textbf{35},  4791--4796  (2002)

\bibitem{Dezin} Dezin,~A.~A.: Multidimensional Analysis and Discrete Models. CRC Press, Boca Raton (1995)

\bibitem{H1} Hestenes D.: Real Spinor Fields.  Journal of Mathematical Physics.
 \textbf{8}, no.~4, 798–-808 (1967)

\bibitem{H2} Hestenes D.: Spacetime Algebra. Gordon and Breach,  New York (1966)

\bibitem{Kahler}
 K\"{a}hler, E.:  Der innere differentialk\"{u}l.  Rendiconti di Matematica \textbf{21}(3--4), 425-523 (1962)

\bibitem{Marchuk}  Marchuk, N. G.:  Dirac-type tensor equations. Nuovo Cimento Soc. Ital. Fis. B(12) 116, no.~11, 1225--1248 (2001)

\bibitem{Rabin}
Rabin, J.M.: Homology theory of lattice fermion doubling.  Nucl. Phys. B. \textbf{201}, no.~2,  315--332 (1982)

\bibitem{S1} Sushch V.: A discrete model of the Dirac-K\"{a}hler equation. Rep. Math. Phys. \textbf{73}, no.~1,  109--125 (2014)
\bibitem{S2} Sushch V.: On the chirality of a discrete Dirac-K\"{a}hler equation.  Rep. Math. Phys. \textbf{76}, no.~2,  179--196 (2015); arXiv:1411.7673 


\end{thebibliography}
\end{document}